\documentclass[centertags,12pt]{amsart}

\usepackage{amssymb}
\usepackage{amsmath}
\usepackage{bm}
\usepackage{amsthm}
\usepackage{color}

\makeatletter
\renewcommand{\subsection}{\@startsection
{subsection}{2}{0mm}{\baselineskip}{-0.25cm}
{\normalfont\normalsize\em}}
\makeatother

\newtheorem{theorem}{Theorem}
\newtheorem{proposition}{Proposition}
\newtheorem{corollary}{Corollary}
\newtheorem{lemma}{Lemma}
{\theoremstyle{definition}

\newtheorem{example}{Example}}
\theoremstyle{remark}
\newtheorem{remark}{Remark}

\begin{document}


\title[LRC codes from curves with separated variables]{Locally Recoverable codes from algebraic curves with separated variables}

\author{Carlos Munuera}
\address{Department of Applied Mathematics, University of Valladolid, Avda Salamanca SN, 47014 Valladolid, Castilla, Spain}
\email{cmunuera@arq.uva.es}

\author{Wanderson Ten\'orio}
\address{Departamento de Matem\' atica, Universidade Federal de Mato Grosso, Av. F. C. Costa~2367, 78060-900, Cuiab\'a, Brazil}
\email{dersonwt@yahoo.com.br}

\author{Fernando Torres} 
\address{Institute of Mathematics, Statistics and Computer Science,  University of Campinas
Cidade Universitaria "Zeferino Vaz", 
Bar\~ao Geraldo 13083-859, Campinas, Brazil}
\email{ftorres@ime.unicamp.br}

\keywords{error-correcting code, locally recoverable code, algebraic geometry code}
\subjclass[2010]{94B27, 11G20, 11T71, 14G50, 94B05}
\thanks{The first author was supported by Spanish Ministerio de Econom\'{\i}a y Competitividad  under grant MTM2015-65764-C3-1-P MINECO/FEDER. The second author was supported by CNPq-Brazil, under grants 201584/2015-8 and 159852/2014-5. The third author was supported by CNPq-Brazil under grant 310623/2017-0.}

\maketitle


\begin{abstract}
A Locally Recoverable code is an error-correcting code  such that any erasure in a single coordinate of a codeword can be recovered from a small subset of other coordinates.  We study  Locally Recoverable Algebraic Geometry codes arising from certain curves defined by equations with separated variables. The recovery of erasures is obtained by means of Lagrangian interpolation in general, and simply by one
addition in some particular cases.
\end{abstract}


\section{Introduction}
\label{section1}

Locally Recoverable  (LRC) codes  were introduced in \cite{GHSY} motivated by the 
use of coding techniques applied to distributed and cloud storage systems. Roughly speaking,  local recovery techniques enable us to repair  lost encoded  data by a local procedure, that is by making use of small amount of data instead of all information contained in a codeword.

Let $\mathcal{C}$ be a linear code of length $n$, dimension $k$ and minimum distance $d$ over the field $\mathbb{F}_{q}$.
A coordinate $i\in\{ 1,\dots,n\}$ is {\em locally recoverable with locality $r$} if there is a {\em recovery set} $R_i\subseteq \{1,\dots,n\}$ with $i\not\in R_i$ and $\# R_i= r$, such that for any two codewords $\mathbf{u},\mathbf{v}\in\mathcal{C}$,  
whenever $\pi_i(\mathbf{u})=\pi_i(\mathbf{v})$ we have $u_i=v_i$, where $\pi_i$ is the projection on the coordinates of $R_i$. Under this condition, an erasure at position $i$ of $\mathbf{v}$  can be  recovered by using the information given by the coordinates of $\mathbf{v}$ with indices in $R_i$. We can also be interested on existence of distinct recovering sets $R_i$ for the same coordinate, which is known as the \emph{availability problem}.
The code $\mathcal{C}$ is {\em locally recoverable with locality} $r$  if any coordinate is locally recoverable with locality at most $r$.

Every code with minimum distance $d>1$ is a LRC code  of locality $r\le k$. In practice we are interested in LRC codes $\mathcal{C}$ allowing small  recovering sets, in relation to the other parameters  $[n, k, d]$ of $\mathcal{C}$. We have the following Singleton-like bound:
the locality $r$ of such a code verifies the relation, \cite{GHSY},
\begin{equation} \label{LRCbound}
k+d+\left\lceil{\frac{k}{r}}\right\rceil\le  n+2
\end{equation}
which gives a lower bound on $r$.
Codes reaching equality are called {\em Singleton-optimal}  (or simply optimal).  The difference $\Delta=(n+2)-(k+d+\left\lceil{\frac{k}{r}}\right\rceil)$  is the {\em Singleton-optimal defect} of $\mathcal{C}$.
 
MDS codes (Reed-Solomon, RS, codes in particular)  are optimal, but they have the largest possible locality $r=k$. In \cite{barg1} a variation of RS codes for local recoverability purposes was introduced by Tamo and Barg. These so-called LRC RS codes are optimal and can have much smaller locality than RS codes.  Its length is smaller than the size of $\mathbb{F}_{q}$.
This is a usual fact:  for most known optimal codes, the cardinality of the ground field $\mathbb{F}_{q}$ is larger than the code length $n$, \cite{largo}. Then the use of such codes for parctical applications rely on alphabets of large size, what limits its usefulness. Thus
the search for long optimal codes has become a challenging problem.
A method to obtain long codes  is to consider codes from algebraic curves with many rational points. In this way the above construction of LRC RS codes was extended by Barg, Tamo and Vladut \cite{barg2}, to the {\em LRC Algebraic Geometry (LRC AG) codes}, obtaining larger LRC codes.  The availability problem for LRC AG codes has been treated in \cite{barg2} and \cite{matthews}. 

In this article we study LRC AG codes coming from curves defined by equations with separated variables $A(Y)=B(X)$, paying special attention to the case in which the degrees of $A(Y)$ and $B(X)$ are coprime. We study also the generalized Hamming weights of these codes, and show how in some special cases the recovery can be done simply by one addition. 
The paper is organized as follows: 
Section 2 contains some introductory material. The core of this article is in Section 3, where the definition and main properties of our codes are treated. The fact that some of these codes admit a simple recovering method by just one addition is stated  in Section \ref{section4}.  The generalized Hamming weights of general LRC codes, and codes coming from curves defined by equations with separated variables in particular, are studied in Section \ref{section5}.  Finally in Section \ref{section6} we show some worked examples of our constructions.

\section{Construction of LRC codes from curves}
\label{section2}

The  construction of LRC codes  we study is based on ideas of  \cite{barg2,Wan}.  For the convenience of the reader, we briefly recall these constructions. Since the codes we study in this article are subcodes of Algebraic Geometry (AG) codes, we begin with a short reminder of this theory. For a complete reference on AG codes we address the reader to \cite{micap,Sti}.

\subsection{Algebraic Geometry codes}
\label{section2.2}

Let $\mathcal{X}$ be a (projective, non-singular, absolutely irreducible, algebraic) curve of genus $g$ defined over the field $\mathbb{F}_q$. Let $\mathcal{P}=\{ P_1,\dots,P_n\}\subseteq\mathcal{X}(\mathbb{F}_q)$ be a set of $n$ rational distinct points, $D=P_1+\dots+P_n$, and let $G$ be a rational divisor with support disjoint from $\mathcal{P}$. The AG code
$\mathcal{C}(\mathcal{P},\mathcal{L}(G))$ is defined as  $\mathcal{C}(\mathcal{P},\mathcal{L}(G))=\mbox{ev}_{\mathcal{P}}(\mathcal{L}(G))$, where $\mathcal{L}(G)=\{ \mbox{rational functions $f$} : \mbox{div}(f)+ G\ge 0\} \cup \{0\}$ is the Riemann space associated to $G$ and $\mbox{ev}_{\mathcal{P}}$ is the evaluation at $\mathcal{P}$ map, $\mbox{ev}_{\mathcal{P}}(f)=( f(P_1),\dots,f(P_n))$.

Given a divisor $E$ we write $\ell(E)=\dim(\mathcal{L}(E))$. For a positive integer $t$ the $t$-th {\em gonality} of $\mathcal{X}$ is defined as $\gamma_t=\min\{ \deg(E) : \mbox{$E$ is a divisor on $\mathcal{X}$ with $\ell(E)\ge t$}\}$.  Thus, from Riemann-Roch theorem we have $\gamma_t\le t+g-1$ with equality if $t\ge 2g-1$. 

The code $\mathcal{C}(\mathcal{P},\mathcal{L}(G))$ is called {\em nonbundant} when $\mbox{ev}_{\mathcal{P}}$ is injective; otherwise 
$\mathcal{C}(\mathcal{P},\mathcal{L}(G))$ is {\em abundant} and $w=\ell(G-D)=\dim(\ker(\mbox{ev}_{\mathcal{P}}))$ is its {\em abundance}. The dimension of this code is $k=\ell(G)-w$ and its minimum distance verifies $d\ge n-\deg(G)+\gamma_{w+1}$ (the generalized Goppa bound).

When the divisor $G$ is a multiple of a single point $G=mQ$, $Q\notin\mathcal{P}$, then the code $\mathcal{C}(\mathcal{P},\mathcal{L}(mQ))$ is called {\em one-point}. The properties and parameters of these codes are closely related to the {\em Weierstrass semigroup} of $\mathcal{X}$ at $Q$, $H=H(Q)=\{ v(f) :$  $f$ is a rational function on $\mathcal{X}$ with poles only at $Q \}$, where $v$ is the pole order at $Q$. Usually we write $H$ as an enumeration of its elements in increasing order, $H=\{ h_1=0<h_2< \dots\}$. It is clear that $\gamma_t\le h_t$ for all $t$,  with equality if $t\ge 2g-1$.

\subsection{LRC codes from Algebraic Geometry}
\label{section2.3}

We can construct  LRC AG codes from algebraic curves as follows \cite{barg2}:  let $\mathcal{X},\mathcal{Y}$ be two algebraic  curves over $\mathbb{F}_{q}$ and let $\phi:\mathcal{X}\rightarrow \mathcal{Y}$ be a rational separable  morphism of degree $r+1$. Take a set $\mathcal{U}\subseteq \mathcal{Y} (\mathbb{F}_{q})$ of rational points with totally split fibres  and let  $\mathcal{P}=\phi^{-1}(\mathcal{U})$. Let $E$ be a rational divisor on $\mathcal{Y}$ with support disjoint  from $\mathcal{U}$ and $\mathcal{L}(E)$ its associated Riemann-Roch space of dimension $m=\ell(E)$. 
By the separability of $\phi$ there exists  $x\in \mathbb{F}_{q}(\mathcal{X})$ satisfying $\mathbb{F}_{q}(\mathcal{X})=\mathbb{F}_{q}(\mathcal{Y})(x)$. Let 

$$
V=\left\{  \sum_{i=0}^{r-1} \; \sum_{j=1}^{m} a_{ij} f_j x^i  \; : \; a_{ij}\in\mathbb{F}_{q}    \right\}
$$
where $\{f_1,\dots,f_m\}$ is a basis of $\mathcal{L}(E)$. The LRC AG code $\mathcal{C}$ is defined as 
$\mathcal{C}=\mbox{ev}_{\mathcal{P}}(V)\subseteq \mathbb{F}_{q}^n$, with $n=\#\mathcal{P}$.
Note that $\mathcal{C}$ is a subcode of  $C(\mathcal{P},\mathcal{L}(G))=\mbox{ev}_{\mathcal{P}}(\mathcal{L}(G))$, where $G$ is any divisor on $\mathcal{X}$ satisfying $V\subseteq \mathcal{L}(G)$. In particular $d(\mathcal{C})\ge d(C(\mathcal{P},\mathcal{L}(G)))\ge n-\deg(G)$.
Let $\mbox{ev}_{\mathcal{P}}(f)$, $f\in V$, be a codeword in $\mathcal{C}$.
Since the functions of $\mathcal{L} (E)$ are constant on each fibre  $\phi^{-1}(U)$, $U\in\mathcal{U}$,  the local recovery of an erased coordinate $f(P)$ of $\mbox{ev}_{\mathcal{P}}(f)$  can be performed by Lagrangian interpolation at the remaining $r$ coordinates of $\mbox{ev}_{\mathcal{P}}(f)$ corresponding to points in the fibre $\phi^{-1}(\phi(P))$ of $P$.  

\begin{theorem}\cite{barg2} 
If $\mbox{ev}_{\mathcal{P}}$ is injective on $V$ then $\mathcal{C}\subseteq \mathbb{F}_{q}^n$ is a linear $[n,k,d]$ LRC code with parameters 
$n=s(r+1)$, $k=r\ell(E)$, $d\geq n-\deg(E)(r+1)-(r-1)\deg(x)$ and locality $r$.
\end{theorem}

\section{LRC codes from  curves with separated variables}
\label{section3}

The above method for constructing LRC codes gives good codes, but it is rather complex technically. In this article
we study that construction for curves defined by equations with separated variables. Furthermore we extend it somewhat in several directions, obtaining more and sometimes better codes. For a reference on the arithmetic and geometry of a class of curves defined by  equations with separated variables we address the reader to \cite[Lemma 6.54]{LN} and \cite[Sect. 12.1]{HKT}.

\subsection{General construction of LRC codes from curves with separated variables}
\label{section3.1}

Let  $A(Y)$, $B(X)$ be two univariate polynomials over $\mathbb{F}_{q}$ of  degrees $a,b$, respectively, and such that 
$A(Y)-B(X)$ is absolutely irreducible.
Let us consider the curve $\mathcal{X}$ of plane affine equation $A(Y)=B(X)$.  As usual, we denote by $x$ and $y$,  respectively, the cosets of $X$ and $Y$ in $\mathbb{F}_{q}[X,Y]/ (A(Y)-B(X))$. 
Throughout this article, we shall assume the following condition 

$(*)$:  the functions $x, y$, have  just one pole, $Q\in\mathcal{X}(\mathbb{F}_{q})$, which is common to both. 

Note that the above condition implies that $\mathcal{X}$ is unibranched at $Q$. This always happens if $Q$ is a regular point. 
In case $Q$ is  singular, then it  it must be a cusp. In any case, the condition $(*)$ implies the existence of a well defined pole-order at $Q$ map on $\mathbb{F}_{q}(\mathcal{X})$, which we will denote by $v$.  Note that  $av(y)=bv(x)$, so when $\gcd(a,b)=1$ then we have $v(x)=a$, $v(y)=b$ and (*) is satisfied.

\begin{remark}\label{ex} 
Two simple cases where the above conditions on $\mathcal{X}$ are verified,  are the following: \newline
(a) if the above plane model is non-singular and $a\neq b$, then $A(Y)-B(X)$ is absolutely irreducible by B\'ezout's theorem, and $Q$ is the only point at infinity of $\mathcal{X}$; \newline 
(b) (\cite[Lemma 12.1]{HKT}) 
Recall that a polynomial is called {\em linearized} if the exponents of monomials associated to all its nonzero coefficients are powers of $p=\mbox{\rm char}(\mathbb{F}_{q})$.
If $B(X)$ is linearized, $B(X)=\lambda_0X+\lambda_1X^p+\lambda_2X^{p^2}+\dots$,  with $\lambda_0\neq 0$, and $A(Y)$ has  degree $a\geq 2$ with $a\not\equiv 0\pmod{p}$, then $A(Y)-B(X)$ is absolutely irreducible and $\gcd(a,b)=1$, so $(*)$ above is satisfied. Moreover, the Weierstrass semigroup at $Q$ is $H=\langle a,b\rangle$ and thus the  genus of $\mathcal X$ is $g=(a-1)(b-1)/2$. In addition, as $B'(X)\neq 0$, we have $\# y^{-1}(\beta)=b$ for any $\beta$ in the algebraic closure of 
${\mathbb F}_q$.   
\end{remark}

Define  the set $\mathcal{L}(\infty Q)=\cup_{m\ge 0} \mathcal{L}(mQ)$. Let us recall that this is a finitely generated $\mathbb{F}_{q}$-algebra.
Take two rational functions $\phi_1,\phi_2\in\mathcal{L}(\infty Q)$. Let $\mathcal{X}(\mathbb{F}_{q})^{+}=\mathcal{X}(\mathbb{F}_{q})\setminus \{Q\}$. Since neither $\phi_1$ nor $\phi_2$ have poles in $\mathcal{X}(\mathbb{F}_{q})^{+}$, we have two well defined maps $\phi_1,\phi_2: \mathcal{X}(\mathbb{F}_{q})^{+} \rightarrow  \mathbb{A}^1(\mathbb{F}_{q})$.
Let $\mathcal{P}_1,\dots,\mathcal{P}_s$ be disjoint subsets of $\mathcal{X}(\mathbb{F}_{q})^{+}$ in which $\phi_1$ is constant. Then   each of these sets, $\mathcal{P}_i$, is contained in a fibre of $\phi_1$, and so $\#\mathcal{P}_i\le v(\phi_1)$. Let $r_i+1$ be the number of different values took by $\phi_2$ when acting over $\mathcal{P}_i$,

$$
r_i=\# \{\phi_2(P)  :  P\in\mathcal{P}_i\}-1, \; i=1,\dots,s.
$$ 
Now fix an integer $r\ge 1$ and select those $\mathcal{P}_i$ verifying $r_i\ge r$  (if any), say $\mathcal{P}_1,\dots,\mathcal{P}_u$. Set $\mathcal{P}=\mathcal{P}_1 \cup \dots \cup \mathcal{P}_u$ and $n=\#  \mathcal{P}$. 
Fix also numbers $\ell_0,\dots\ell_{r-1}$ and consider  the linear space of functions
\begin{equation}\label{defV}
V= \sum_{i=0}^{r-1} \epsilon_i\langle 1, \phi_1,\dots,\phi_1^{\ell_i} \rangle \phi_2^i \subset \mathcal{L}(\infty Q)
\end{equation}
where $\epsilon_i=0$ or $1$ and $\langle 1, \phi_1,\dots,\phi_1^{\ell_i} \rangle$ stands for the linear space generated by $ 1, \phi_1,\dots,\phi_1^{\ell_i}$ over $\mathbb{F}_{q}$. Then we get a code
$\mathcal{C}=\mathcal{C}(\mathcal{P},V)=\mbox{ev}_{\mathcal{P}}(V)$.  As in the case of AG codes we say that $\mathcal{C}(\mathcal{P},V)$ is abundant if the evaluation map related to $V$ is not injective.
In order to give an estimate of its parameters  we shall consider the number
\begin{equation}\label{defm}
m=m(V)= \max\{\epsilon_i(\ell_iv(\phi_1)+iv(\phi_2)) : i=0,\dots,r-1\}.
\end{equation}

\begin{lemma}\label{subcodeAG}
$V\subseteq \mathcal{L}(mQ)$ and so $\mathcal{C}(\mathcal{P},V)$ is a subcode of the algebraic geometry code  $\mathcal{C}(\mathcal{P},\mathcal{L}(mQ))$.
\end{lemma}
\begin{proof}
By the properties of valuations, if $f\in V$ then $v(f)\le m$ and so $f\in\mathcal{L}(mQ)$.
\end{proof}

To be consistent with the usual notation used in algebraic geometry, we shall write $\ell (V) = \dim (V)$. Furthermore, for a divisor $G$ on $\mathcal{X}$ we define $\mathcal{L}_V(G)=V\cap \mathcal{L}(G)$.

\begin{theorem} \label{construccionLRC}
$\mathcal{C}(\mathcal{P},V)$ is a $[n,k,d]$ LRC code of locality $r$ with $k=\ell(V)-\dim(\mathcal{L}_V(mQ-D))\ge \ell(V)-w$ and $d\ge d(\mathcal{C}(\mathcal{P},\mathcal{L}(mQ)))\ge n-m+\gamma_{w+1}$ where $w=\ell(mQ-D)$ and $\gamma_{w+1}$ is the $(w+1)$-th gonality of $\mathcal{X}$. In particular, 
if $m<n$ then $w=0$ hence $k=\ell(V)$ and $d\ge n-m$. 
\end{theorem}
\begin{proof}
The kernel of the evaluation map $\mbox{ev}_{\mathcal{P}}:V\rightarrow\mathbb{F}_{q}^n$ is $\mathcal{L}_V(mQ-D)$.
In particular  when $m<n$ this map is injective. This facts, together with Lemma \ref{subcodeAG} and the generalized Goppa bound on the minimum distance, imply the statements about $k$ and $d$.
Let us see that  $\mathcal{C}(\mathcal{P},V)$ is a LRC code of locality $r$ whose recovery sets are subsets of $\mathcal{P}_1,\dots,\mathcal{P}_u$.  Let $f\in V$ and suppose we want to recover an erasure at position $P$, $P\in\mathcal{P}_i$. Let  $\{  P_{i,1},\dots,P_{i,r+1}=P\}\subseteq \mathcal{P}_i$ be a set in which $\phi_2$ takes $r+1$ different values. Since all functions in $\mathbb{F}_{q}[\phi_1]$ are constant on  $\mathcal{P}_i$,  the restriction of $f$ to this set acts as a polynomial $L_i=\sum_{j=0}^{r-1} a_j T^j $ of degree $\le r-1$,  that is $f(P_{i,j})=L_i(\phi_{2}(P_{i,j}))$ for all $j=1,\dots,r+1$. Since  $\phi_{2}$ takes $r$ different values in 
$\{ P_{i,1},\dots,P_{i,r}\}$, the polynomial $L_i$ may be computed by Lagrangian interpolation from $\phi_{2}(P_{i,1}),\dots,\phi_{2}(P_{i,r})$ and $f(P_{i,1}),\dots,f(P_{i,r})$. Finally $f(P)=L_i(\phi_2(P))$.
\end{proof}

\begin{example}\label{exvarios}
(a) (Example 1 of \cite{barg1}). Consider the curve $Y=X^3$ over $\mathbb{F}_{13}$. This is a rational curve with 13 affine points plus one point $Q$ at infinity. Let $\phi_1=y, \phi_2=x$. The fibres of $\phi_1$ are the sets $\mathcal{P}_1=\{ 1,3,9\}$, $\mathcal{P}_2=\{ 2,6,5\}$ and $\mathcal{P}_3=\{4,10,12\}$ (note that $13\equiv 1 \; (\mbox{mod } 3)$ hence  $\mathbb{F}_{13}$ contains a cubic root of unity). We obtain optimal LRC codes of length $9$, locality $2$ and dimensions $k=2,4,6$. \newline
(b) Let us slightly modify the curve of (a) to get larger codes. The curve $Y^2=X^3$ over $\mathbb{F}_{13}$ is again rational and hence it has 13 affine points. Take $\xi=2$ as a primitive element of $\mathbb{F}_{13}$.
The curve $\mathcal{X}:Y^2=X^3+2$ is nonsingular of genus 1. 
It has 18 affine points. Since $13\equiv 1 (\mbox{mod } 3)$,  $\omega=\xi^4=3$ is a cubic root of unity. Then the map $\sigma(x,y)=(\omega x,y)$ is an automorphism of $\mathcal{X}$, whose orbits have length 3. Thus the 18 points of 
$\mathcal{X}(\mathbb{F}_{13})^+$ are grouped in 6 orbits $\mathcal{P}_{\beta}=\{ (\alpha,\beta), (\omega\alpha,\beta), (\omega^2\alpha,\beta )\}$.  By taking $\phi_1=y,\phi_2=x$, we get LRC optimal codes of length $n=18$, locality $r=2$ and dimensions $k=3,5,7,9,11$.\newline
(c) LRC codes arising from the Hermitian curve $\mathcal{H}:Y^{q+1}=X^q+X$  over $\mathbb{F}_{q^2}$ were treated in \cite{barg2}, although in that work the authors only consider the case in which $m<n$ and $\ell_0=\dots=\ell_{r-1}$. To give a concrete example, take $q=4, \phi_1=y, \phi_2=x$. Let $Q$ be the only point at infinity of $\mathcal{H}$. We obtain codes of length $n=64=\# \mathcal{H}(\mathbb{F}_{16})^+$. Since $v(x)=5, v(y)=4$ we have $r=3$. The biggest code we get comes from the space

$$
V= \sum_{i=0}^{2} \langle 1, y,\dots,y^{13} \rangle x^i \subseteq\mathcal{L}(62 Q)
$$
which leads to a code of dimension $k=42$ and minimum distance $d\ge n-62=2$. A simple computation shows that such code has optimal defect $\Delta\le 8$. \newline
(d) Following the procedure stated above, we can obtain bigger and better codes than in \cite{barg2}. By continuing with the example of the Hermitian curve $\mathcal{H}$ over  $\mathbb{F}_{16}$, we can consider the space

$$
V= \sum_{i=0}^{2} \langle 1, y,\dots,y^{16-i} \rangle x^i. 
$$
Then $\ell(V)=48$ and $m(V)=66$. The evaluation map is not injective as $f=\prod_{\beta\in\mathbb{F}_{16}}(y-\beta)\in V$ and $f(P)=0$ for every point $P\in\mathcal{H}(\mathbb{F}_{16})^+$, hence $\mathcal{C}(\mathcal{P},V)$ is an abundant code. Since the second gonality of $\mathcal{H}$ is known to be $\gamma_2=q=4$  (eg. \cite{Castle}) we have $w=\ell(\mathcal{L}(66Q-D))=1$ and hence the kernel of the evaluation map related to $V$ is generated by $f$. Then we obtain a code of dimension $47$ and distance $d\ge n-66+\gamma_2=2$. Its optimal  defect is $\Delta\le 1$ (compare to (c)). \newline
(e) LRC codes from the Norm-Trace curve have been studied in \cite{ballico} where the approach is similar to that of (c). 
\end{example}

\subsection{The case of prime degrees}
\label{section3.2}

The bound for the minimum distance of $\mathcal{C}(\mathcal{P},V)$ given in Theorem \ref{construccionLRC}  strongly depends on the values $v(\phi_1)$ and $v(\phi_2)$. This, and the previous examples,  invites us to consider functions $\phi_1,\phi_2$ for which these values are as small as possible. Of particular interest is the case in which the degrees $a$ and $b$ are coprime $\gcd(a, b) = 1$. In this setting $v(x)=a, v(y)=b$ and the Weiestrass semigroup at $Q$ contains $\langle a,b \rangle$,  $H\supseteq\langle a,b \rangle$, where $Q$ is the common pole of $x$ and $y$.  This suggests taking $\phi_1=y$ and $\phi_2=x$ (or vice versa). 
This is the case treated in this subsection, and from now on in this article.

Any element $m\in \langle a,b \rangle$ can be written uniquely as $m =\lambda a + \mu b$ with $\lambda <b$. Thus we have
\begin{equation}\label{sumaL}
\mathcal{L}(\infty Q)\supseteq\mathbb{F}_{q} [x,y]= \bigoplus_{i=0}^{b-1} \mathbb{F}_{q}[y]x^i.
\end{equation}

Let $\phi_1=y$ and $\phi_2=x$.
Let $\mathcal{U}$ be the set of unramified points of $\phi_1$, that is the set of $\beta\in\mathbb{F}_{q}$ such that  the fibre $\phi_1^{-1}(\beta)$ totally split and so it consists of $b=v(\phi_1)$ distinct rational points in $\mathcal{X}(\mathbb{F}_{q})^+$. Let $u=\# \mathcal{U}$ and $n=ub$. For $\beta\in\mathcal{U}$ set $\mathcal{P}_\beta=\phi_1^{-1} (\beta)=\{ P_{\beta,1},\dots,P_{\beta,b}\}$.
Write $D_\beta= P_{\beta,1}+\dots+P_{\beta,b}$,  $\mathcal{P}=\cup_{\beta\in\mathcal{U}}\mathcal{P}_\beta$ and $D=\sum_{\beta\in\mathcal{U}}D_\beta$. 

In the best case, $\mathcal{U}=\mathbb{F}_{q}$ and hence $\mathcal{P}=\mathcal{X}(\mathbb{F}_{q})^+$. This is the case of the so-called {\em Castle curves}, see \cite{Castle}. Note that if this holds then the number rational of points of $\mathcal{X}$ attains the Lewittes bound $\# \mathcal{X}(\mathbb{F}_{q})\le qh_2+1$. 

\begin{proposition}
(a)  $\mbox{\rm div} (\phi_1-\beta)=D_{\beta}-bQ$ hence we have the equivalence of divisors $D_\beta\sim bQ$. \newline
(b)  $\mbox{\rm div} (\prod_{\beta\in \mathcal{U}}(\phi_1-\beta))=D-nQ$ hence  $D\sim nQ$.
\end{proposition}

Being $\mathcal{X}$ a plane curve, the function $\phi_2=x$ takes $b$ different values on each $\mathcal{P}_\beta$, $\beta\in\mathcal{U}$. Set $r=b-1$. As in the former Section \ref{section3.1} we consider the linear space of functions $V \subset  \mathbb{F}_{q}[x,y] \subseteq\mathcal{L}(\infty Q)$ stated by equation (\ref{defV}).

\begin{proposition}
If $\phi_1=y, \phi_2=x$ and $\gcd(a,b)=1$, then the sum defining $V$ in equation (\ref{defV}) is  direct, 
$V= \bigoplus_{i=0}^{b-2} \epsilon_i \langle 1, \phi_1,\dots,\phi_1^{\ell_i} \rangle \phi_2^i$,
hence $\ell(V)=\epsilon_0(1+\ell_0)+\dots+\epsilon_{b-2}(1+\ell_{b-2})$.
\end{proposition}
\begin{proof}
If $\langle 1, \phi_1,\dots,\phi_1^{\ell_i} \rangle \phi_2^i \cap \langle 1, \phi_1,\dots,\phi_1^{\ell_j} \rangle \phi_2^j\neq (0)$ for some $i\le j$, then by taking pole orders we get $(\ell_i-\ell_j)v(\phi_1)=(j-i)v(\phi_2)$. Thus $v(\phi_1)=b| (j-i)<b$ so $i=j$.
\end{proof}

In our case $m=m(V)=\max\{\epsilon_i(\ell_ib+ia) : i=0,\dots,b-2\}$.
The semigroup  $H=H(Q)= \{h_1=0<h_2<\dots\}$ can help us to simplify our estimates on the parameters of  $\mathcal{C}(\mathcal{P},V)$. For  a non negative integer $m$ define

$$
\iota(m)=\max\{ t : h_t \le m\}.
$$
Then $\iota(m)=\ell(mQ)$, hence from Riemann-Roch theorem we have $\iota(m)\ge m+1-g$ with equality if $m\ge 2g$, where $g$ is the genus of $\mathcal{X}$ (or equivalently the genus of $H$). Since $\gamma_t\le h_t$ for all $t$ and $D\sim nQ$, we have $w=\ell((m-n)Q)=\iota(m-n)$. According to  Theorem \ref{construccionLRC},  $\mathcal{C}(\mathcal{P},V)$ is an LRC code of locality $r=b-1$, length $n$, dimension $k=\ell(V)-\dim(\mathcal{L}_V(mQ-D))\ge \ell(V)-\iota(m-n)$  and minimum distance $d\ge n-m+\gamma_{\iota(m-n)+1}$. A upper bound on $d$ is given by the following result.

\begin{proposition}\label{cotasup}
Assume $\epsilon_0=1$.
If $\ell_0b \ge n$ then $\mathcal{C}(\mathcal{P},V)$ is an abundant code. 
If $\ell_0b<n$ then the minimum distance of $\mathcal{C}(\mathcal{P},V)$ verifies $n-m\le d\le n-\ell_0b$. 
\end{proposition}
\begin{proof}
Let $\tau=\prod_{\beta\in \mathcal{U}^*}(\phi_1-\beta) \in \langle 1, \phi_1,\dots,\phi_1^{\ell_0} \rangle \subseteq V$, where $\mathcal{U}^*=\mathcal{U}$ if $\ell_0b \ge n$ and $\mathcal{U}^*$ is a subset of $\ell_0$ distinct elements of $\mathcal{U}$ if $\ell_0b<n$. Then  $\mbox{ev}_{\mathcal{P}}(\tau)$  vanishes at all points of $\mathcal{P}_{\beta}$, $\beta\in\mathcal{U}^*$. This means that $\tau\in\ker(\mbox{ev}_{\mathcal{P}})$ if $\ell_0b \ge n$ and $d\le \mbox{wt}(\mbox{ev}_{\mathcal{P}}(\tau))=n-\ell_0b$ if $\ell_0b<n$.
\end{proof}

\subsection{Complete spaces}
\label{section3.3}

To conclude this section, we will study which linear spaces $V$ produce LRC codes with the best possible parameters in the case $\gcd(a,b)=1, \phi_1=y,\phi_2=x$.
Given a positive integer $m$ we can consider the space

$$
V_m= \mathcal{L}(mQ) \cap \bigoplus_{i=0}^{b-2} \mathbb{F}_{q}[y]x^i.
$$
Thus $m(V_m)\le m$ and equality holds when $m\in \langle a,b\rangle$. If $V$ is a linear space of functions defined by an equation (\ref{defV}) then $V\subseteq V_{m(V)}$, $m(V)=m(V_{m(V)})$, and the Goppa bounds on minimum distances of  $\mathcal{C}(\mathcal{P},V)$ and $\mathcal{C}(\mathcal{P},V_{m(V)})$ given by Theorem \ref{construccionLRC} coincide. We  say that $V_{m(V)}$ is the {\em completion} of $V$, and $V$ is {\em complete} if $V=V_{m(V)}$. 
Let us detail a little more the structure of these $V$. From (\ref{sumaL}), 
\begin{equation} \label{oplusL}
\mathcal{L}(mQ) \cap  \mathbb{F}_{q}[x,y] =\bigoplus_{i=0}^{b-1} \epsilon_i  \langle 1, y,\dots,y^{\ell_i} \rangle x^i
\end{equation}
with $\epsilon_i=0$ if $ia>m$ and $\epsilon_i=1$, $\ell_ib+ia\le m$ if $ia\le m$. Then $\ell_i=\lfloor (m-ia)/b\rfloor$ for  $ia\le m$. 
Thus we have the following.

\begin{proposition}\label{Vcomplete}
$V$ is complete iff  $\epsilon_i=1$ for $i=0,\dots,s$, where $s=\min\{b-2,\lfloor m(V)/a\rfloor\}$, $\epsilon_i=0$ otherwise, and  $\ell_i=\lfloor (m(V)-ia)/b\rfloor$, $i=0,\dots,s$. 
\end{proposition}

\begin{corollary}
If $V$ is complete then the minimum distance of $\mathcal{C}(\mathcal{P},V)$ verifies 
$n-m(V)\le d\le n-\lfloor m(V)/b\rfloor b$.
\end{corollary}
\begin{proof}
If $V$ is complete then $\ell_0=\lfloor m(V)/b\rfloor$ and the result follows from Proposition \ref{cotasup}  and the Goppa bound.
\end{proof}

In some cases the set $V$ is precisely a Riemann space $\mathcal{L}(mQ)$. If this happens then  $\mathcal{C}(\mathcal{P}, V)$  is an AG code and so its dimension is maximized with respect to the Goppa bound we have used to estimate its minimum distance. An obvious necessary (but not sufficient) condition for this to hold is that $V$ be complete.

\begin{proposition}\label{VAG}
Let $V$ be a linear space of functions defined as in equation (\ref{defV}) and let $m=m(V)$. Then $V=\mathcal{L}(mQ)$ if and only if 
$H=\langle a,b \rangle$, $V$ is complete and $m<a(b-1)$.
\end{proposition}
\begin{proof}
If $V$ is not complete, then $V\neq\mathcal{L}(mQ)$ by definition. Assume $V$ is complete and $H=\langle a,b \rangle$. According to Proposition \ref{Vcomplete} we have equality $V=\mathcal{L}(mQ)$ if and only if $\epsilon_{b-1}=0$ in equation (\ref{oplusL}).  This happens when $m<a(b-1)$.
\end{proof}

\begin{corollary}\label{VAGcor}
Let $V$ be a complete linear space of functions  and let $m=m(V)$.  If $H=\langle a,b \rangle$ and $m \ge a(b-1)$ then
$\ell(mQ)-\ell(V) = 1+\lfloor (m-a(b-1))/b\rfloor$.
\end{corollary}
\begin{proof}
If  $m \ge a(b-1)$ then $\ell(mQ)-\ell(V) =\dim( \langle 1, y,\dots,y^{\ell_{b-1}} \rangle x^{b-1})=1+\ell_{b-1}$.
\end{proof}

This result allows us to give a bound on the optimal defect of codes coming from complete spaces. For simplicity we restrict to the non-abundant case $m<n$.

\begin{corollary}\label{OptDef}
Let $V$ be a complete linear space of functions. If $H=\langle a,b \rangle$ and  $m=m(V)<n$, then the Singleton-optimal defect $\Delta$ of  $\mathcal{C}(\mathcal{P},V)$ verifies \newline
(a) If $m<a(b-1)$ then $\Delta\le g+1-\lceil (m+1-g)/(b-1)\rceil\le g$; \newline
(b) If $m\ge a(b-1)$ then $\Delta\le g+2+\ell_{b-1}-\lceil (m-g-\ell_{b-1})/(b-1)\rceil\le g+\ell_{b-1}+1$; \newline
where $g$ is the genus of $\mathcal{X}$ and $\ell_{b-1}=\lfloor (m-a(b-1))/b\rfloor$.
\end{corollary}
\begin{proof}
Let $[n,k,d]$ be the parameters of $\mathcal{C}(\mathcal{P},V)$ and let $[n,k^*,d^*]$ be the parameters of $\mathcal{C}(\mathcal{P},\mathcal{L}(mQ))$. Then $d\ge d^*$, $k^*\ge m+1-g$ and $k=k^*$ if $m<a(b-1)$, $k=k^*-1-\ell_{b-1}$ if $m\ge a(b-1)$. The result follows from a straightforward computation using the well known fact that $k^*+d^*\ge n+1-g$ (see \cite[Sect. 4.2]{micap}).
\end{proof}

\begin{example} \label{elipDelta}
Let us consider the elliptic curve $\mathcal{X}:Y^2=X^3+2$ over $\mathbb{F}_{13}$ of Example \ref{exvarios}(b). 
Here $g=1, a=2, b=3, H=\langle 2,3 \rangle$. Take the space of functions $V=\langle 1,y,y^2\rangle \oplus \langle 1,y,y^2\rangle x$. This is a complete space of dimension $l(V)=6$ and $m(V)=8$. Note that

$$
\mathcal{L}(8Q)=\langle 1,y,y^2\rangle \oplus \langle 1,y,y^2\rangle x \oplus \langle 1,y\rangle x^2
$$
so $\ell_{b-1}=\ell_2=1$.  As explained in Example \ref{exvarios}(b),  we a get a code  $\mathcal{C}(\mathcal{P},V)$ of length $n=18$, dimension $k=6$ and locality $r=2$. According to Corollary \ref{OptDef}(b), its optimal defect verifies $\Delta\le 1$. A direct computation shows that this estimate gives the true value of $\Delta$. 
\end{example}

\section{A simplified recovering method}
\label{section4}

In this section we shall show how in some cases the recovering process can be performed simply by one 
addition, which is much faster and simpler than the interpolation method. Working with curves with separated variables is fundamental to the approach we present, which is based on ideas of \cite{Wan}. We keep the same notations 
and assumptions of the previous section. In particular we assume $\phi_1=y$ and $\phi_2=x, \gcd(a,b)=1$. 

Let $L(T)=T^s+\lambda_{s-1}T^{s-1}+\dots+\lambda_0\in\mathbb{F}_{q}[T]$ be a polynomial. The roots $a_1,\dots,a_s$ of $L$ are related to its coefficients by the {\em Vieta's formulae} \cite{Vieta}, $\sigma_i=(-1)^i \lambda_{s-i}$, $1\le i\le s$, where $\sigma_i=\sigma_i(a_1,\dots,a_s)$  is the $i$-th elementary symmetric polynomial on  $a_1,\dots,a_s$, that is $\sigma_1=a_1+\dots+a_s; \dots ; \sigma_s =  a_1\cdots a_s$. For $i\ge 1$, the sums of successive $i$-powers of  $a_1,\dots,a_s$, 
$\pi_i=\pi_i(a_1,\dots,a_s)=a_1^i+\dots+a_s^i$, are related to the elementary symmetric polynomials by the \emph{Newton-Girard indentities} \cite{Vieta}:  $\pi_1=\sigma_1$ and  for each integer $i>1$, 
\begin{equation}\label{NG}
\pi_i  = (-1)^{i-1}i\sigma_i  -\sum_{j=1}^{i-1} (-1)^j\pi_{i-j}\sigma_j .
\end{equation}
Therefore the sums of powers of roots of $L(T)$ are related to its coefficients through \eqref{NG}. Such relations can be applied to simplify the recovering method in our codes from the curve $A(Y)=B(X)$. 
Let $V=\bigoplus_{i=0}^{r-1} \epsilon_i\langle 1,\phi_1,\dots,\phi_1^{\ell_i}\rangle 
\phi_2^i$.

\begin{theorem}\label{easyrec}
If $B(X)$ is a polynomial of degree $b\geq 3$ with $\pi_i=0$ for $1\le i\le b-2$, and any of the following conditions: 
\begin{enumerate}
\item [(i)] $\mbox{\rm char}(\mathbb{F}_{q}) | b$; or
\item [(ii)] $\epsilon_0=0$;
\end{enumerate}
is verified, then for each fibre $\mathcal{P}_{\beta}=\{ P_{\beta,1},\dots, P_{\beta,b}\}$ of $\phi_1=y$, $\beta\in\mathcal{U}$, and each function $f\in V$, it holds that $\sum_{i=1}^{b} f( P_{\beta,i})=0$. Thus the recovering of one erasure can be obtained by one addition. 
\end{theorem}
\begin{proof}
Let $\beta\in\mathcal{U}$ and let $f=\sum_{j=0}^{b-2} g_i(y)x^i \in V$.
The points of $\mathcal{P}_{\beta}$ have coordinates $(\alpha_1,\beta),\dots, (\alpha_b,\beta)$, being $\alpha_1,\dots,\alpha_b$, the roots of the polynomial $B(X)-A(\beta)$. From \eqref{NG} the sums of $i$-powers of roots of $B(T)-A(\beta)$ coincide with the sums of $i$-powers of roots of $B(T)$ for $i=1,\ldots,b-2$, since the coefficients of $B(T)$ and $B(T)-A(\beta)$ are equal except by the constant part $A(\beta)$. Then

$$
\sum_{i=1}^{b} f( P_{\beta,i})= \sum_{i=1}^{b}\sum_{j=0}^{b-2} g_j(\beta)\alpha_i^j = 
\sum_{j=0}^{b-2} g_j(\beta)\pi_j(\alpha_1,\ldots,\alpha_b) = b g_0(\beta)=0
$$
by the conditions (i) and (ii).
\end{proof}

\begin{example}
The class of polynomials $B(X)=X^b+\lambda_1X+\lambda_0$ over $\mathbb{F}_{q}$ fits in the framework of the previous Theorem. Indeed, $\pi_1=\sigma_1=-\lambda_{b-1}=0$. By induction, if $\pi_1=\dots=\pi_{i-1}=0$ for $1<i\le b-2$, then from equation (\ref{NG}) we have $\pi_i=(-1)^{i-1} i\sigma_i=- i\lambda_{b-i}=0$. Therefore $\pi_1=\dots=\pi_{b-2}=0$. For instance, in the Hermitian curve $\mathcal{H}$ of Example~\ref{exvarios}(c)(d) we have $B(X)=X^q+X$. Hence the conditions of the above theorem are fulfilled and the LRC codes in that example have recovering obtained by one addition.
\end{example}

\begin{example}
Linearized polynomials over $\mathbb{F}_q$ constitute also a class of polynomials in the conditions of Theorem~\ref{easyrec}.  Indeed, assuming that $B(X)$ is a linearized polynomial of degree $b$, a power of $\mbox{\rm char}(\mathbb{F}_{q})$, it can be proved by induction that  $\pi_1=\dots=\pi_{b-2}=0$, \cite{Wan}. For instance, the Norm-Trace curves of Example~\ref{exvarios}(e) have $B(X)$ given by a linearized polynomial and their associated LRC codes have recovering obtained by one addition. More in general, LRC codes arising from Artin-Schreier curves under our construction, admit a recovering performed by one addition.
\end{example}

\section{Generalized Hamming weights and optimal rank of LRC codes}
\label{section5}

An important invariant associated to a $[n,k]$ code $\mathcal{C}$ is its {\em weight hierarchy}, that is the sequence of its generalized Hamming weights $d_1,\dots,d_k$, where 

$$
d_t=\min\{ \# \mbox{supp}(\mathcal{E}) : \mbox{$\mathcal{E}$ is a $t$-dimensional linear subcode of $\mathcal{C}$ }\},
$$ 
$ \mbox{supp}(\mathcal{E}) =\cup _{\mathbf{v}\in \mathcal{E}} \mbox{supp}(\mathbf{v})$ and  the {\em support} of a vector $\mathbf{v}$ is the set of positions $i$ where $v_i\neq 0$. In particular $d_1$ is the usual minimum distance, see \cite[Sect. 4.5.1]{Pel}. 

In this section we shall show an extension of the bound given by equation (\ref{LRCbound}) to all generalized weights. This extension is valid for any LRC code and the proof is similar to that of (\ref{LRCbound}) given  in \cite{GHSY}. Next we  will focus on our codes from curves with separate variables. We recall that Hamming weights of LRC codes from algebraic geometry have already been studied in  \cite{ballico}.

\subsection{A bound on the Generalized Hamming weights of LRC codes}
\label{section5.1}

Let $\mathcal{C}$ be a $[n,k,d]$ nondegenerate LRC code. It is well known that the $t$-th generalized Hamming weight of $\mathcal{C}$ verifies  (see \cite[Prop. 4.3.12]{Pel})

\begin{equation} \label{SlikeGeneralizado}
n-d_t=\max\{ \# R\subseteq\{1,\dots,n\} :  \dim(\mathcal{C}(R))\le k-t \}.
\end{equation}
For $i=1,\dots,n$, let $R_i$ be a minimal recovery set for coordinate $i$, and write $\bar{R}_i=R_i\cup\{i\}$.

\begin{theorem}  \label{PesGen}
Let $\mathcal{C}$ be a $[n,k]$ nondegenerate LRC code of locality $r$. For $t=1,\dots,k$, it holds that
\begin{equation}\label{PesGenCota}
k+d_t+\left\lceil \frac{k-t+1}{r} \right\rceil  \le  n+t+1.
\end{equation}
\end{theorem} 
\begin{proof}
Starting from $S_0=\emptyset$, iteratively construct sets $S_1,\dots,S_l \subseteq\{1,\dots,n\}$, in the following way: while $\dim(\mathcal{C}(S_{i-1}))<k-t$, choose an index ${j_i}$ such that $\dim(\mathcal{C}(S_{i-1}\cup R_{j_i}))>\dim(\mathcal{C}(S_{i-1}))$  and define $S_i=S_{i-1}\cup \bar{R}_{j_i}$. Then $l$ is the smallest index such that $\dim(\mathcal{C}(S_l))\ge k-t$. By a similar argumet as in \cite[Theorem 5]{GHSY}, we have
 $\dim(\mathcal{C}(\bar{R}_{i}))<\# \bar{R}_i\le r+1$, so  $\#S_l\ge \dim(\mathcal{C}(S_l))+l$ and $l\ge \dim(\mathcal{C}(S_l))/r$. Thus
 
$$
\# S_l\ge \dim(\mathcal{C}(S_l))+\frac{ \dim(\mathcal{C}(S_l))}{r} .
$$
Let $\dim(\mathcal{C}(S_l))=k-t+\delta$ for some $0\le \delta<r$. If $\delta=0$ set $S=S_l$. We get $\# S\ge k-t+(k-t)/r\ge k-t-1+(k-t+1)/r$.
If $\delta>0$, since $R_{j_l}$ is minimal  we can remove $\le \delta+1$ coordinates in $\bar{R}_{j_l}$ to obtain a set $R'_l$ such that   $S=\bar{R}_{j_1}\cup \dots \cup \bar{R}_{j_{l-1}} \cup R'_{l}$ verifies $\dim(\mathcal{C}(S))=\dim(\mathcal{C}(S_l))-\delta=k-t$. As $\# S\ge \# S_l-\delta-1$, substituting in the above  equation we obtain
$\# S \ge k-t-1+(k-t+\delta)/r \ge  k-t-1+(k-t+1)/r$.
Then the result follows by applying (\ref{SlikeGeneralizado}) to the set $S$.
\end{proof}

The bound of Theorem \ref{PesGen} extends both the Singleton optimal bound given in equation (\ref{LRCbound})  and the Singleton bound for generalized Hamming weights, $k+d_t\le n+t$. Let $\mathcal{C}$ be a LRC code of length $n$. The smallest $t$ for which we have equality in this bound,  $k+d_{t}= n+{t}$ is called the {\em MDS rank} of $\mathcal{C}$, $t=\mbox{mdsrank}(\mathcal{C})$. Similarly we can define the {\em Singleton optimal rank} of $\mathcal{C}$, $\mbox{optrank}(\mathcal{C})$, as the  smallest $t$ for which we have equality in the bound (\ref{PesGenCota}). Thus the optimal rank of $\mathcal{C}$ measures how far from optimal is the code $\mathcal{C}$. Contrary to that happens for the MDS rank,  having  equality in this bound for a certain $t$ does not imply  equality for all $t'> t$ (see Example \ref{elip} below).

\begin{proposition}\label{mdsrank}
Let $\mathcal{C}$ be a code of length $n$, locality $r$ and dimension $k$. Then we have \newline
(a) $\mbox{\rm optrank}(\mathcal{C})\le \mbox{\rm mdsrank}(\mathcal{C})$.\newline
(b) $\mbox{\rm mdsrank}(\mathcal{C})\ge k-r+1$. \newline
(c) $d_{k-r}\le n-r-1$. If equality holds then $\mbox{\rm optrank}(\mathcal{C})\le k-r$.
\end{proposition}
\begin{proof}
(a) If $d_t=n-k+t$ then we have equality in the bound  (\ref{PesGenCota}). (b) Note that   
$\mathcal{C}$ has MDS rank $t$ iff for any set $R$ of $k-t+1$ coordinates,  $\mathcal{C}(R)$ has full rank. Then $k+t+1\le r$.
(c) Take a minimal recovering set $R$. From  equation (\ref{SlikeGeneralizado}) we have $n-d_{k-r}\ge r+1$. If $n-d_{k-r}= r+1$ then we have equality in (\ref{PesGenCota}) for $t=k-r$.
\end{proof}

\subsection{Generalized Hamming weights of LRC codes from curves with separated variables}
\label{section5.3}

Let us return to our case of LRC codes from the curve $\mathcal{X}:A(Y)=B(X)$ with $\gcd (a,b)=1$. Keeping the notation of the previous sections, let $V= \bigoplus_{i=0}^{r-1} \epsilon_i \langle 1, \phi_1,\dots,\phi_1^{\ell_i} \rangle \phi_2^i$ and $m=m(V)$. As in the case of the minimum distance, we have (see \cite{micap})
\begin{equation}\label{n-dt}
n-d_t=\max\{ \deg (E) : E\le D, \dim (\mathcal{L}_V(mQ-E))\ge t+w\}
\end{equation}
so $d_t\ge n-m+\gamma_{w+t}$. In some cases this bound provides us the true value of $d_t$.

\begin{proposition} \label{Dbeta}
Let $\mu\le u$ be a positive integer such that $\mu b<m$ and let $t=\iota(m-\mu b)-w$. Then $d_t\le n-\mu b$.
\end{proposition}
\begin{proof}
Take $\beta_1,\dots,\beta_u\in\mathcal{U}$ and let $E=D_{\beta_1}+\dots+D_{\beta_\mu}$. Then $E\sim \mu b Q$ hence
$\ell (mQ-E)=\ell ((m-\mu b)Q)=\iota(m-\mu b)\ge t+w$. According to (\ref{n-dt}) we have $n-d_t\ge \deg (E)=\mu b$.
\end{proof}

\begin{example} \label{elip}
Let us consider the curve $\mathcal{X}:Y^2=X^3+2$ over $\mathbb{F}_{13}$ of examples \ref{exvarios}(b) and \ref{elipDelta}. Take the space of functions $V=\langle 1,y,y^2\rangle \oplus \langle 1,y,y^2\rangle x$. We a get a code of length $n=18$, dimension $k=6$ and locality $r=2$.  The gonality sequence of $\mathcal{X}$ is $0,2,3,\dots$, and $m(V)=8$.  The bounds on its generalized Hamming weights given by equations (\ref{PesGenCota}) and  (\ref{n-dt})   are listed in the following Table \ref{Table1}, where the  true values are indicated in boldface. These are obtained as follows:
by taking a function $f=(y-\beta_1)(y-\beta_2)(x-\alpha)\in V$ (for example $f=(y-1)(y+1)(x-1)$) we get a codeword of weight 10. By Proposition  \ref{Dbeta} we have $d_2=12$ and $d_5=15$, so $d_3=13, d_4=14$. Finally, since $1\in V$ we have $d_6=18$. 
{\small
\begin{table}[htbp]
\begin{center}
\begin{tabular}{c | cccccc}
$t$                                             & 1           & 2           & 3           & 4           & 5            & 6 \\ \hline 
lower bound (\ref{n-dt})               & {\bf 10}  & {\bf 12}  & {\bf 13}  & {\bf 14}  & {\bf 15}   & 16 \\
upper bound  (\ref{PesGenCota}) & 11         & {\bf 12}  & 14         & 15         & 17          & {\bf 18} \\  \hline
\end{tabular}
\vspace*{1mm}
\caption{Hamming weights of Example \ref{elip}.}
\label{Table1}
\end{center}
\end{table}}   
Note that the optimal rank of $\mathcal{C}(\mathcal{P},V)$ is 2, but we have not equality in the  bound (\ref{PesGenCota}) for $t=3,4,5$. Furthermore, its MDS rank is the one estimated by Proposition \ref{mdsrank}(b).
\end{example}

\begin{proposition}
Let $V=V_m$ be  complete space with $m<n$. If $H(Q)=\langle a,b\rangle$, then for $t=1,\dots,k=\dim (\mathcal{C}(\mathcal{P}, V))$ we have \newline
(a) If $m<a(b-1)$ then $d_t(\mathcal{C}(\mathcal{P}, V))=d_t(\mathcal{C}(\mathcal{P}, \mathcal{L}(mQ)))$. \newline
(b) If $m\ge a(b-1)$ then
$d_t(\mathcal{C}(\mathcal{P}, \mathcal{L}(mQ))) \le d_t(\mathcal{C}(\mathcal{P}, V)) \le d_{t+\ell_{b-1}+1}(\mathcal{C}(\mathcal{P},  \mathcal{L}(mQ)))$,
where $\ell_{b-1}$ is defined as in Section \ref{section3.3}, $\ell_{b-1}=\lfloor (m-a(b-1))/b\rfloor$.
\end{proposition}
\begin{proof}
According to Proposition \ref{VAG} and Corollary \ref{VAGcor}, 
if $m<a(b-1)$ then $\mathcal{C}(\mathcal{P}, V))=\mathcal{C}(\mathcal{P}, \mathcal{L}(mQ)$; if $m\ge a(b-1)$ then
$\mathcal{C}(\mathcal{P}, V))$ is a linear subspace of $\mathcal{C}(\mathcal{P},  \mathcal{L}(mQ))$ of co-dimension $\ell_{b-1}+1$, hence for any subspace $W\subseteq \mathcal{C}(\mathcal{P},  \mathcal{L}(mQ))$ of dimension $t+\ell_{b-1}+1$ we have
$\dim(\mathcal{C}(\mathcal{P}, V))\cap W)\ge t$.
\end{proof}

\section{Two worked examples}
\label{section6}

To end this article we present examples of two families of curves and the corresponding LRC codes arising from them. To obtain concrete numerical results, we will detail each of the examples over the field $\mathbb{F}_{64}$.

\subsection{The  Kondo-Katagiri-Ogihara curve}
\label{section6.1}

Let us consider the curve $\mathcal{X}:Y^{q}+Y=X^{q^s+1}$ over the field $\mathbb{F}_{q^{2s}}$, where $q$ is a prime power and $s$ and odd integer. Codes arising from this curve have been studied in \cite{kondo,kondo2}.
 $\mathcal{X}$ has one singular point at infinity, $Q$, plus $q^{2s+1}$ affine points. Its genus is $g=q^s(q-1)/2$, hence it is a maximal curve and the Weierstrass semigroup of $Q$ is $H=\langle q, q^s+1\rangle$.
 
Let $\xi$ be a primitive element of $\mathbb{F}_{q^{2s}}$.  Then $\xi^{q^s-1}$ is a primitive $(q^s+1)$-th root of unity in 
$\mathbb{F}_{q^{2s}}$ and so the map $\sigma(x,y)=(\xi^{q^s-1}x,y)$ is an automorphism of  $\mathcal{X}$. The orbits of $\mathcal{X}(\mathbb{F}_{q^{2s}})^+$ under the action of $\sigma$ are as follows:  there are $q(q^s-1)$ orbits with $q^s+1$ points; and $q$ orbits with one point, $(0,\beta)$. Then, by taking the sets $\mathcal{P}_{\beta}$ as these multi-point  orbits, $\mathcal{P}_{\beta}=\{ (\alpha,\beta), (\xi^{q^s-1}\alpha,\beta),\dots \}$,  we get LRC codes from $\mathcal{X}$ of length $n=q^{2s+1}-q$ and locality $r=q^s$.

\begin{example}\label{final1}
Take $q=2, s=3$ and consider the curve $\mathcal{X}:Y^2+Y=X^9$ over $\mathbb{F}_{64}$. This is a maximal hyperelliptic curve of genus $g=4$. It has 128 rational affine points plus one point at infinity, $Q$, which is the only rational hyperelliptic point of $\mathcal{X}$. The Weierstrass semigroup of $Q$ is $H=\langle 2,9\rangle$. Observe that, according to Clifford's theorem, for all non-negative integer $t$ the $t$-th gonality of $\mathcal{X}$ is precisely the $t$-th element of $H$, $\gamma_t=h_t$. Let $\xi$ be a primitive element of $\mathbb{F}_{64}$. Then $\xi^7$ is a primitive ninth root of unity. Under the action of  the automorphism $\sigma(x,y)=(\xi^{7}x,y)$, the 128 rational points of  $\mathcal{X}(\mathbb{F}_{64})^+$ are grouped in 14 orbits of length 9, namely

$$
\begin{array}{rlcrl}
\mathcal{P}_{1}=&\{ (\xi, \xi^{18}),\dots,(\xi^{57},\xi^{18}) \}, & &
\mathcal{P}_{2}=&\{ (\xi, \xi^{54}),\dots,(\xi^{57},\xi^{54}) \}, \\
\mathcal{P}_{3}=&\{ (\xi^2, \xi^{36}),\dots,(\xi^{58},\xi^{36}) \}, & &
\mathcal{P}_{4}=&\{ (\xi^2, \xi^{45}),\dots,(\xi^{58},\xi^{45}) \}, \\
\mathcal{P}_{5}=&\{ (\xi^3, \xi^{31}),\dots,(\xi^{59},\xi^{31}) \}, &&
\mathcal{P}_{6}=&\{ (\xi^3, \xi^{59}),\dots,(\xi^{59},\xi^{59}) \}, \\
\mathcal{P}_{7}=&\{ (\xi^4, \xi^{9}),\dots,(\xi^{60},\xi^{9}) \}, &&
\mathcal{P}_{8}=&\{ (\xi^4, \xi^{27}),\dots,(\xi^{60},\xi^{27}) \}, \\
\mathcal{P}_{9}=&\{ (\xi^5, \xi^{47}),\dots,(\xi^{61},\xi^{47}) \}, & &
\mathcal{P}_{10}=&\{ (\xi^5, \xi^{61}),\dots,(\xi^{61},\xi^{61}) \}, \\
\mathcal{P}_{11}=&\{ (\xi^6, \xi^{55}),\dots,(\xi^{62},\xi^{55}) \}, & &
\mathcal{P}_{12}=&\{ (\xi^6, \xi^{62}),\dots,(\xi^{62},\xi^{62}) \}, \\
\mathcal{P}_{13}=&\{ (\xi^7, \xi^{21}),\dots,(1,\xi^{21}) \}, & &
\mathcal{P}_{14}=&\{ (\xi^7, \xi^{42}),\dots,(1,\xi^{42}) \},
\end{array}
$$
plus 2 orbits of length one, $\{(0,0)\}, \{(0,1)\}$. Let $\mathcal{P}=\mathcal{P}_1\cup\dots\cup\mathcal{P}_{14}$, $\phi_1=y,\phi_2=x$, and hence $a=2,b=9$. We construct LRC codes of length $n=126$ and locality $r=8$ by evaluating at $\mathcal{P}$ the functions in the linear spaces 
$
V=\bigoplus_{i=0}^7 \epsilon_i  \langle1,y,\dots,y^{\ell_i} \rangle x^i .
$
To give a concrete example, take $m=50$ and let $V$ the complete space $V=V_{50}$.  As computed in Proposition \ref{Vcomplete}, we have $\ell_0= \ell_1= \ell_2=5 ; \ell_3= \ell_4= \ell_5= \ell_6= \ell_7=4$, so $k=\dim(\mathcal{C}(\mathcal{P},V))=\ell(V)=43$. The minimum distance of this code is at least $d\ge 126-50=76$, and thus its optimal defect is $\Delta\le 3$.  Its generalized Hamming weights can be bounded by using the relations given  by equations (\ref{PesGenCota}) and  (\ref{n-dt}). For example, the first five of them verify
$76 \le d_1 \le 79, 78 \le d_2 \le 80 \le d_3 \le 81,82 \le d_4 < d_5=84$. 
\end{example}

\subsection{A quotient of the Hermitian curve}
\label{section6.2}

Let $\mathcal{X}$ by the curve over $\mathbb{F}_{q^2}$ defined by the equation $:Y^{s}=X^{q}+X$, where $q$ is a prime power and $s| q+1$.  These curves have been studied in \cite[Sect. VI.4]{Sti} for example. $\mathcal{X}$ has genus $(s-1)(q-1)/2$ and $q(1+(q-1)s)$ rational affine points, plus one point at infinity. Then it is a maximal curve. Note that the mapping $\alpha\mapsto \alpha^q+\alpha$ is the trace map onto $\mathbb{F}_{q}$. Let $\xi$ be a primitive element of $\mathbb{F}_{q^2}$. Then $\omega=\xi^{(q+1)/s}$ is a primitive $(q-1)s$-th root of unity. 
Let $\mathcal{U}^*$ be the subgroup of $\mathbb{F}_{q^2}^*$ generated by $\omega$ and $\mathcal{U}=\mathcal{U}^*\cup \{0\}$. If $\beta\in\mathcal{U}$ then $\beta^s\in\mathbb{F}_{q}$ and so the polynomial $T^q+T=\beta^s$ has $q$ roots in $\mathbb{F}_{q^2}$. 

Let $\phi_1=y,\phi_2=x$,  $\mathcal{P}_\beta=\phi_1^{-1}(\beta)$ for $\beta\in\mathcal{U}$ and $\mathcal{P}=\cup_{\beta\in\mathcal{U}}\mathcal{P}_\beta$. We obtain LRC codes of length $n=q(1+(q-1)s)$ and locality $r=q-1$. Furthermore, according to Theorem \ref{easyrec}, when $q\ge 3$ then the recovery of one erasure can be obtained by one addition.

\begin{example}\label{final2}
Take $q=8, s=3$ and consider the curve $\mathcal{X}:Y^{3}=X^{8}+X$ over $\mathbb{F}_{64}$. It has genus 7 and 176 affine points grouped in 22 sets $\mathcal{P}_\beta$ of 8 points. By considering the spaces of functions
$V=\bigoplus_{i=0}^6 \epsilon_i  \langle1,y,\dots,y^{\ell_i} \rangle x^i$,
we obtain LRC codes of locality $r=7$. For example, if $V=V_{50}$ then $\ell_0=6, \ell_1=\ell_2=\ell_3=5, \ell_4=\ell_5=\ell_6=4$. The corresponding code has dimension 40, minimum distance $\ge 126$ and optimal defect $\le 6$. 
This defect is worse than that of Example \ref{final1}, which can be compensated by the fact that the recovery or an erasure is carried out  by a simple addition.
\end{example}

Conversely, since $\omega^{q-1}$ is a primitive $s$-th root of unity, for any $\alpha\in\mathbb{F}_{q^2}$, the polynomial $T^s=\alpha^q+\alpha$ has one root if $\alpha^q+\alpha=0$ and $s$ roots otherwise, all of them belonging to the set $\mathcal{U}$. Therefore we can group the $sq(q-1)$ points $(\alpha,\beta)\in\mathcal{X}(\mathbb{F}_{q^2})$ with $\alpha^q+\alpha=0$ in $q(q-1)$ fibres $\phi_2^{-1}(\alpha)$, each of them with $s$ points. In this way we get  LRC codes of length $n=sq(q-1)$ and locality $r=s-1$. 

\begin{example}
Let us consider again the curve $\mathcal{X}:Y^{3}=X^{8}+X$ over $\mathbb{F}_{64}$ of Example \ref{final2}. 
It has 168 rational affine points  $(\alpha,\beta)$ with $\alpha^8+\alpha \neq 0$. They are grouped in 56 fibers  $\phi_2^{-1}(\alpha)$,  with $3$ points each. By taking spaces 
$V= \epsilon_0  \langle1,x,\dots,x^{\ell_0} \rangle  \oplus \epsilon_1  \langle 1,x,\dots,x^{\ell_1}\rangle y$,
we obtain LRC codes of length $n=168$ and locality $r=2$. For example, if $V=V_{50}=V_{48}$ then $\ell_0=16, \ell_1=14$. The corresponding code has dimension 32, minimum distance $\ge 120$ and optimal defect $\le 2$. 
\end{example}

\section*{Acknowledgments}

The third author wishes to thank the research group SINGACOM from Valladolid University for the finantial support received during his academic visit in January-February 2018.

\end{document}